\documentclass[conference,onecolumn]{IEEEtran}

\usepackage{graphicx}
\usepackage{commath}
\usepackage{epstopdf}
\usepackage{psfrag}
\usepackage{verbatim}
\usepackage{amsfonts}
\usepackage{amssymb}
\usepackage{amsmath}
\usepackage{amsthm}
\usepackage{mathrsfs}

\usepackage{xspace}
\usepackage{bbm}
\input{mathlig}

\newcommand{\muspace}{\mspace{1mu}}

\DeclareRobustCommand{\scond}{\mathchoice{\muspace\vert\muspace}{\vert}{\vert}{\vert}}
\mathlig{|}{\scond}

\newcommand{\Cc}{\mathcal{C}}

\newcommand{\Wc}{\mathcal{W}}
\newcommand{\Xc}{\mathcal{X}}
\newcommand{\Yc}{\mathcal{Y}}
\newcommand{\Zc}{\mathcal{Z}}

\newcommand{\aep}{{\mathcal{T}_{\epsilon}^{(n)}}}

\newcommand{\Wt}{{\tilde{W}}}

\newcommand{\Yt}{{\tilde{Y}}}

\newcommand{\wt}{{\tilde{w}}}

\newcommand{\yt}{{\tilde{y}}}

\def\a{\alpha}
\def\b{\beta}

\def\e{\epsilon}

\DeclareMathOperator\E{\text{\sffamily\upshape E}}
\let\P\relax
\DeclareMathOperator\P{\text{\sffamily\upshape P}}

\newcommand{\U}{\mathrm{Unif}}

\def\textiid{i.i.d.\@\xspace}
\newcommand\iid{\ifmmode\text{ i.i.d. } \else \textiid \fi}

\def\mathllap{\mathpalette\mathllapinternal}
\def\mathllapinternal#1#2{%
  \llap{$\mathsurround=0pt#1{#2}$}}

\def\clap#1{\hbox to 0pt{\hss#1\hss}}
\def\mathclap{\mathpalette\mathclapinternal}
\def\mathclapinternal#1#2{%
  \clap{$\mathsurround=0pt#1{#2}$}}

\let\oldstackrel\stackrel
\renewcommand{\stackrel}[2]{\oldstackrel{\mathclap{#1}}{#2}}

\newcommand{\Gbar}{G\mathllap{\overline{\vphantom{G}\hphantom{\rule{6.25pt}{0pt}}}\mspace{1.1mu}}}
\renewcommand{\hbar}{h\mathllap{\overline{\vphantom{h}\hphantom{\rule{4.6pt}{0pt}}}\mspace{0.77mu}}}

\catcode`~=11 
\newcommand{\urltilde}{\kern -.06em\lower .3em\hbox{~}\kern .02em}
\catcode`~=13 

\newcommand{\bern}{\mathrm{Bern}}

\newcommand{\one}{\mathbf{1}}
\newcommand{\real}{\mathrm{R}}

\newcommand{\xb}{\bar{x}}
\newcommand{\yb}{\bar{y}}
\newcommand{\zb}{\bar{z}}
\newcommand{\pb}{\bar{p}}

\begin{document}
\newtheorem{definition}{Definition}
\newtheorem{theorem}{Theorem}
\newtheorem{corollary}{Corollary}
\newtheorem{lemma}{Lemma}
\newtheorem{problem}{Problem}
\newtheorem{proposition}{Proposition}
\newenvironment{claimproof}[1]{\par\noindent\underline{Proof:}\space#1}{\hfill $\blacksquare$}
\newtheorem{claim}{Claim}
\IEEEoverridecommandlockouts

\title{Exact Common Information}
\author{
    \IEEEauthorblockN{Gowtham Ramani Kumar}
    \IEEEauthorblockA{Electrical Engineering\\
    Stanford University\\
    Email: gowthamr@stanford.edu}

    \and

    \IEEEauthorblockN{Cheuk Ting Li}
    \IEEEauthorblockA{Electrical Engineering\\
    Stanford University\\
    Email: ctli@stanford.edu}

    \and

    \IEEEauthorblockN{Abbas El Gamal}
    \IEEEauthorblockA{Electrical Engineering\\
    Stanford University\\
    Email: abbas@stanford.edu}
    \thanks{ This work was partially supported by Air Force grant FA9550-10-1-0124.   }
}

\maketitle

\thispagestyle{plain}
\pagestyle{plain}

\begin{abstract}
This paper introduces the notion of exact common information, which is the minimum description length of the common randomness needed for the exact distributed generation of two correlated random variables $(X,Y)$. We introduce the quantity $G(X;Y)=\min_{X\to W \to Y} H(W)$ as a natural bound on the exact common information and study its properties and computation. We then introduce the exact common information rate, which is the minimum description rate of the common randomness for the exact generation of a 2-DMS $(X,Y)$. We give a multiletter characterization for it as the limit $\Gbar(X;Y)=\lim_{n\to \infty}(1/n)G(X^n;Y^n)$. While in general $\Gbar(X;Y)$ is greater than or equal to the Wyner common information, we show that they are equal for the Symmetric Binary Erasure Source. We  do not know, however, if the exact common information rate has a single letter characterization in general.

\end{abstract}

\IEEEpeerreviewmaketitle
\section{Introduction}

What is the common information between two correlated random variables or sources? This is a fundamental question in information theory with applications ranging from distributed generation of correlated sources~\cite{Wyner1975a} and secret keys~\cite{Ahlswede--Csiszar1993} to joint source channel coding~\cite{Cover--El-Gamal--Salehi1980}, among others. One of the most studied notions of common information is due to Wyner~\cite{Wyner1975a}. Let $(\Xc\times \Yc, p(x,y))$ be a 2-DMS (or correlated sources $(X,Y)$ in short). The Wyner common information $J(X;Y)$ between the sources $X$ and $Y$ is the minimum common randomness rate needed to generate $(X,Y)$ with asymptotically vanishing total variation. Wyner established the single-letter characterization
\[
J(X;Y) = \min_{W:\, X \to W \to Y} I(W;X,Y).
\]
In this paper we introduce the notion of {\em exact common information}, which is closely related in its operational definition to the Wyner common information. While the Wyner setup assumes block codes and {\em approximate} generation of the 2-DMS $(X,Y)$, our setting assumes variable length codes and {\em exact} generation of $(X,Y)$. As such, the relationship between our setup and Wyner's is akin to that between the zero-error and the lossless source coding problems. In the source coding problem the entropy of the source is the limit on both the zero-error and the lossless compression. Is the limit on the exact common information rate the same as the Wyner common information? We show that they are the same for the Symmetric Binary Erasure Source (SBES) as defined in Section~\ref{sec:1-shot}. We do not, however, know  if they are equal in general.

The rest of this paper is organized as follows. In the next section we introduce the exact distributed generation problem and define the exact common information. We introduce the ``common-entropy" quantity $G(X;Y)=\min_{X\to W \to Y} H(W)$ as a natural bound on the exact common information and study some of its properties. In Section~\ref{sec:definition}, we define the exact common information rate for a 2-DMS. We show that it is equal to the limit $\Gbar(X;Y)=(1/n)G(X^n;Y^n)$ and that it is in general greater than or equal to the Wyner common information. One of the main results in this paper is to show that $\Gbar(X;Y) =J(X;Y)$ for the SBES. A consequence of this result is that the quantity $G(X^k;Y^k)$ can be strictly smaller than $kG(X;Y)$, that is, the per-letter common entropy can be reduced by increasing the dimension. We then introduce the notion of approximate common information rate, which relaxes the condition of exact generation to asymptotically vanishing total variation and show that it is equal to the Wyner common information. As computing the quantity $G(X;Y)$ involves solving a non-convex optimization problem,  in Section~\ref{sec:computingG} we present cardinality bounds on $W$ and use them to find an explicit expression for $G(X;Y)$ when $X$ and $Y$ are binary. Due to space limitation, we do not include many of the proofs. We also mention a connection to the matrix factorization problem in machine learning that would be interesting to explore further.

\section{Definitions and Properties}\label{sec:1-shot}
Consider the distributed generation setup depicted in Figure~\ref{fig:oneshotproblem}. Alice and Bob both have access to common randomness $W$. Alice uses $W$ and her own local randomness to generate $X$ and Bob uses $W$ and his own local randomness to generate $Y$ such that $(X,Y) \sim p_{X,Y}(x,y)$. We wish to find the limit on the least amount of common randomness needed to generate $(X,Y)$ exactly.
   \begin{figure}[h]
    	\begin{center}
     	\psfrag{w}[b]{$W$}
     	\psfrag{T1}[r]{Alice}
    	\psfrag{T2}[l]{Bob}
     	\psfrag{a}[c]{Decoder 1}	
    	\psfrag{b}[c]{Decoder 2}
    	\psfrag{u}[t]{$\hat X$}
    	\psfrag{v}[t]{\,$\hat Y$}
    	\includegraphics[scale=0.55]{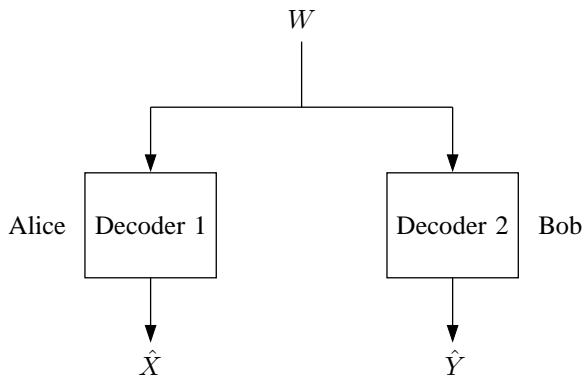}
	\vspace{5pt}
        \caption{Setting for distributed generation of correlated random variables. For exact generation, $(\hat X,\hat Y) \sim p_{X,Y}(x,y)$.}
        \label{fig:oneshotproblem}
    	\end{center}
    \end{figure}

    More formally, we define a {\em simulation code} $(W,R)$ for this setup to consist of
    \begin{itemize}
    \item[$\circ$] A common random variable $W \sim p_{W}(w)$. As a measure of the amount of common randomness, we use the per-letter {\em minimum expected codeword length} $R$ over the set of all variable length {\em prefix-free} zero-error binary codes $\Cc \subset \{0,1\}^*$ for $W$, i.e., $R= \min_{\Cc} \E(L)$, where $L$ is the codeword length of the code $\Cc$ for $W$.
    \item[$\circ$] A stochastic decoder $p_{\hat X|W}(x|w)$ for Alice and a stochastic decoder $p_{\hat Y|W}(y|w)$ for Bob such that $\hat X$ and $\hat Y$ are conditionally independent given $W$.
    \end{itemize}
The random variable pair $(X,Y)$ is said to be exactly generated by the simulation code $(W,R)$ if $p_{\hat X,\hat Y}(x,y)=  p_{X,Y}(x,y)$.
We wish to find the {\em exact common information} $R^*$ between the sources $X$ and $Y$, which is the infimum over all rates $R$ such that the random variable pair $(X,Y)$ can be exactly generated.

Define the following quantity, which can be interpreted as the ``common entropy" between $X$ and $Y$,
\begin{equation}\label{eqn:defineH}
        G(X;Y)=  \min_{W:\,X\to W\to Y} H(W).
\end{equation}
\smallskip
    \noindent{\bf Remark}: We can use $\min$ instead of $\inf$ in the definition of $G(X;Y)$ because the cardinality of $W$ is bounded as we will see in Proposition~\ref{prop:cardbnd}, hence the optimization for computing $G(X;Y)$ is over a closed set.

Following the proof of Shannon's zero-error compression theorem, we can readily show the following.
\begin{proposition}\label{1-shot}
\[
G(X;Y) \le R^* < G(X;Y) + 1.
\]
\end{proposition}

Computing $G(X;Y)$ is in general quite difficult (see Section~\ref{sec:computingG}). In some special cases, we can find an explicit expression for it.
\smallskip

\noindent{\bf Example 1} The Symmetric Binary Erasure Source (SBES) $(X,Y)$ is defined by
    \begin{align*}
        X&\sim \bern(1/2), \\
        Y&=\begin{cases}
                X & \text{w.p. } 1-p,\\
                e & \text{w.p. } p,
            \end{cases}
    \end{align*}
    where $p$ is the {\em erasure probability} for the source. It can be shown that for the SBES,
    \[
    G(X;Y)=\min \{1,H(p)+1-p\}.
    \]
    Note that the Wyner common information for this source is~\cite{channelsynth}
    \[
        J(X;Y) =
        \begin{cases}
            1 & \text{if }p\le 0.5,\\
            H(p) &     \text{if }p> 0.5.
        \end{cases}
    \]

In the following we present some basic properties of $G(X;Y)$.
\subsection{Properties of $G(X;Y)$}
    \begin{enumerate}
    \item $G(X;Y)\ge 0$ with equality if and only if $X$ and $Y$ are independent.
    \begin{proof}
        First assume $G(X;Y)=0$. Suppose $W$ achieves $G(X;Y)$. Then $X \to W \to Y$ and $H(W)=0$. Thus $W=\phi$, constant. Hence $X,Y$ are independent.

        To show the converse, if $X,Y$ are independent, $X\to \phi \to Y$ and $G(X;Y)\le H(\phi)=0$. Thus $G(X;Y)=0$.
    \end{proof}
    \item$G(X;Y)\ge J(X;Y)$.
    \begin{proof}
        \begin{align*}
            G(X;Y)&=\min_{W:X\to W\to Y} H(W) \\
            &= H(W*) \\
            &\ge I(W*;X,Y) \\
            &\ge \min_{W:X\to W\to Y} I(W;X,Y)\\
            &=J(X;Y).
        \end{align*}
    \end{proof}

    \item {\em Data-processing Inequality:} If $U\to X\to Y$ forms a Markov chain, then $G(U;Y)\le G(X;Y)$.
        \begin{proof}
            Let $W$ achieve $G(X;Y)$. Then $U\to X\to W \to Y$ forms a Markov chain. Hence, $G(U;Y)\le H(W)=G(X;Y)$.
        \end{proof}
    \item\label{itm:conditionalH1} Define $G(X;Y|Z)=\sum_{z \in \Zc} p_Z(z)G(X;Y|Z=z)$. Then $G(X;Y)\le H(Z)+G(X;Y|Z)$.
        \begin{proof}
            For each $Z=z$, choose $W_z$ as the random variable that achieves $G(X;Y|Z=z)$. Then $X\to (Z,W_Z)\to Y$ forms a Markov chain. Therefore,
            \begin{align*}
                G(X;Y) &\le H(Z,W_Z) \\
                &=H(Z)+H(W_Z|Z)\\
                &=H(Z)+\sum_z p_Z(z) H(W_z) \\
                &=H(Z)+G(X;Y|Z).
            \end{align*}
        \end{proof}
    \item If there exist functions $f(X)$ and $g(Y)$ such that $Z=f(X)=g(Y)$, then $G(X;Y) = H(Z)+G(X;Y|Z)$.
        \begin{proof} First observe that if $Z$ satisfies the condition $Z=f(X)=g(Y)$ and $W$ satisfies the Markov condition $X\to W \to Y$, then $Z$ is a function of $W$.
        To see why, note that \[p_{X,Y}(x,y)=\sum_{w\in \Wc}p_W(w)p_{X|W}(x|w)p_{Y|W}(y|w).\]
        Therefore, if $p_{X,Y}(x,y)=0$ and $p_{X|W}(x|w)>0$, then $p_{Y|W}(y|w)=0$.

        Now we will show that if $p_{X|W}(x_1|w)>0$ and $p_{X|W}(x_2|w)>0$, then $f(x_1)=f(x_2)$. If not, for any $y$ such that $g(y)\neq f(x_1)$, $p_{X,Y}(x_1,y)=0$, therefore $p_{X|W}(x_1|w)>0\implies p_{Y|W}(y|w)=0$. Similarly for any $y$ such that $g(y)\neq f(x_2)$, $p_{X,Y}(x_2,y)=0$, therefore, $p_{X|W}(x_2|w)>0\implies p_{Y|W}(y|w)=0$. As a consequence, for any $y$, $p_{Y|W}(y|w)=0$, a contradiction.

        Thus, for any $w$ the set $\{f(x):p_{X|W}(x|w)>0\}$ has exactly one element. If we define $h(w)$ as that unique element, then $f(X)=h(W)$. Therefore, $Z=f(X)=g(Y)=h(W)$.

            To complete the proof, let $W$ achieve $G(X;Y)$. Then,
            \begin{align*}
                G(X;Y)&=H(W)=H(W,Z)\\
                &=H(Z)+H(W|Z)\\
                &=H(Z)+\sum_z p_Z(z)H(W|Z=z)\\
                &\ge H(Z)+\sum_z p_Z(z) G(X;Y|Z=z)\\
                &= H(Z)+ G(X;Y|Z).
            \end{align*}
            This, in combination with Property~\ref{itm:conditionalH1}, completes the proof.
        \end{proof}
    \item\label{prop:minsuff} Let $T(X)$ be a sufficient statistic of $X$ with respect to $Y$ ~(\cite{minimalsufficientstatistic}, pg. 305). Then $G(X;Y)=G(T(X);Y)$. Further, if $W$ achieves $G(X;Y)$, we have $H(W)\le H(T(X))$. Thus a noisy description of $X$ via $W$ may potentially have a smaller entropy than the minimal sufficient statistic, which is a deterministic description.
        \begin{proof}
            Observe that both Markov chains $T(X)\to X\to Y$ and $X\to T(X)\to Y$ hold. Hence by the data-processing inequality, $G(X;Y)=G(T(X);Y)$. Now, since $W$ achieves $G(X;Y)$ and $T(X)\to T(X)\to Y$,
            \[
                H(W)=G(X;Y)=G(T(X);Y)\le H(T(X)).
            \]
        \end{proof}
    \end{enumerate}
\section{Exact Common Information Rate}\label{sec:definition}

The distributed generation setup in Figure~\ref{fig:oneshotproblem} can be readily extended to the $n$-letter setting in which Alice wishes to generate $X^n$ from common randomness $W_n$ and her local randomness and Bob wishes to generate $Y^n$ from $W_n$ and his local randomness such that $p_{\hat X^n,\hat Y^n} (x^n,y^n) \sim \prod_{i=1}^n p_{X,Y}(x_i,y_i)$. We define a simulation code $(W_n,R,n)$ for this setup in the same manner as for the one-shot case.

We say that Alice and Bob can exactly generate the 2-DMS $(X,Y)$ at rate $R$ if for some $n \ge 1$, there exists a $(W_n, R,n)$ simulation code that exactly generates $(X^n,Y^n)$ (since we assume prefix-free codes for $W_n$, we can simulate for arbitrarily large lengths via concatenation of successive codewords). We wish to find the {\em exact common information rate} $R^*$ between the sources $X$ and $Y$, which is the infimum over all rates $R$ such that the 2-DMS $(X,Y)$ can be exactly generated.

Define the ``joint common entropy"
\begin{equation}\label{eqn:defineHn}
        G(X^n;Y^n)=  \min_{W_n:\,X^n\to W_n\to Y^n} H(W_n).
\end{equation}
It can be readily shown that $\lim_{n\to \infty} (1/n)G(X^n;Y^n)=\inf_{n \in \mathbb{N} } (1/n)G(X^n;Y^n)$. Hence, we can define the limiting quantity
\[
   \Gbar(X;Y)=\lim_{n\to \infty} \frac{1}{n}G(X^n;Y^n).
\]
	We are now are ready to establish the following multiletter characterization for the exact common information rate.   	
    \begin{proposition}[Multiletter Characterization of $R^*$]\label{prop:asymptotic}
        The exact common information rate between the components $X$ and $Y$ of a 2-DMS $(X,Y)$ is
		\[
			R^* = \Gbar(X;Y). 
		\]
    \end{proposition}
    \begin{proof} 
        \emph{Achievability:} Suppose $R>\Gbar(X;Y)$. We will show that the rate $R$ is achievable.

         Since $\Gbar(X;Y)=\lim_{n\to\infty}(1/n)G(X^n;Y^n)$, $R\ge (1/n) (G(X^n;Y^n) + 1)$ for $n$ large enough. By the achievability part of Shannon's zero-error source coding theorem, it is possible to exactly generate $(X^n,Y^n)$ at rate at most $(1/n)\left(G(X^n;Y^n) + 1\right)$. Hence rate $R$ is achievable and thus $R^*\le \Gbar(X;Y)$.
		
		\emph{Converse:} Now suppose a rate $R$ is achievable. Then there exists a $(W_n,R,n)$- simulation code that exactly generates $(X^n,Y^n)$. Therefore, by the converse for Shannon's zero-error source coding theorem, $R\ge (1/n) G(X^n;Y^n)$. Since $(1/n)G(X^n;Y^n)\ge \Gbar(X;Y)$, we conclude that $R^*\ge \Gbar(X;Y)$.
    \end{proof}

As expected the exact common information rate is greater than or equal to the Wyner common information.
\begin{proposition}\label{prop:bound}
\[
\Gbar(X;Y) \ge J(X;Y).
\]
\end{proposition}
The proof of this result is in Appendix~\ref{proof:greater}.

In the following section, we show that they are equal for the SBES in Example 1. We do not know if this is the case in general, however.

\subsection{Exact Common Information of The SBES}

    We will need the following result regarding computing the Wyner common information for the SBES.
		\begin{lemma}\label{clm:sbesWchoice}
            To compute $J(X;Y)$ for the SBES, it suffices to consider $W$ of the form
            \[
                W=\begin{cases}
                X & \text{ w.p. }1-p_1,\\
                e & \text{ w.p. }p_1,
                \end{cases}
            \]
            and
            \[
                Y=\begin{cases}
                W & \text{ w.p. }1-p_2,\\
                e & \text{ w.p. }p_2,
                \end{cases}
            \]
            where $p_1,p_2$ satisfy $p_1+p_2-p_1p_2=p$, the erasure probability of the SBES.
        \end{lemma}
        The proof follows by~\cite{channelsynth}, Appendix A.

     We now present the main result on exact common information rate in this paper.
    \begin{theorem}\label{thm:SBES}
        If $(X,Y)$ is an SBES, then $\Gbar(X;Y)=J(X;Y)$.
    \end{theorem}
	\begin{proof}
       In general $\Gbar(X;Y)\ge J(X;Y)$. We will now provide an achievability scheme to show that for SBES, $\Gbar(X;Y)\le J(X;Y)$.

         Choose a $W$ as defined in Lemma \ref{clm:sbesWchoice} and define
        \begin{align*}
            \tilde{W}=\begin{cases}
            d & \text{ if }W\in \{0,1\},\\
            e & \text{ if }W=e,
            \end{cases} \\
            \tilde{Y}=\begin{cases}
            d & \text{ if }Y\in \{0,1\},\\
            e & \text{ if }Y=e.
            \end{cases}
        \end{align*}
        Note that $\Yt^n$, denoting the location of the erasures, is i.i.d. $\bern(p)$ (with $1 \leftarrow e$, $0 \leftarrow d$) and independent of $X^n$. Furthermore, $Y^n$ is a function of $X^n$ and $\Yt^n$.
        \smallskip

        \noindent{\emph{Codebook Generation:}}
        Generate a codebook $\Cc$ consisting of $2^{n(I(\Yt;\Wt)+\e)}$ sequences $\wt^n(m)$, $m\in [1: 2^{n(I(\Yt;\Wt)+\e)}]$, that ``covers" almost all the $\yt^n$ sequences except for a subset of small probability $\delta(\e)$. By the covering lemma~(\cite{abbasbook}, page 62), such a codebook exists for large enough $n$.

        This lets us associate every covered sequence $\yt^n$ with a unique $\wt^n=\wt^n(\yt^n)\in \Cc$ such that $(\yt^n,\wt^n)\in \aep$.

        Define the random variable
        \begin{equation}\label{eqn:wtn}
            \Wt_n=\begin{cases}
                \wt^n(\yt^n)& \text{if $\yt^n$ is covered by $\Cc$}, \\
                \yt^n& \text{if $\yt^n$ is not covered}.
            \end{cases}
        \end{equation}
        Note that $\Wt_n$ is a function of $\Yt^n$ and that the set of erasure coordinates in $\Wt_n$ is a subset of those in $\Yt^n$.
		
	\noindent{\emph{Channel Simulation Scheme:}}
        \begin{enumerate}
            \item The central node generates $\Wt_n$ defined in \eqref{eqn:wtn} and sends it to both encoders.
            \item Encoder 2 (Bob) generates $\Yt^n \sim p_{\Yt^n|\Wt_n}(\yt^n|\wt^n)$ 
            \item The central node generates and sends to both encoders a message $M$ comprising i.i.d. $\bern(1/2)$ bits for only those coordinates $i$ of $X^n$ where $\Wt_n(i)=d$. Thus $H(M)\le n(1-p_1+\delta(\e))$.
            \item Encoder 1 (Alice) generates the remaining bits of $X^n$ not conveyed by $M$ using local randomness. Then $X^n$ is independent of $\Wt_n,\Yt^n$ and is i.i.d. $\bern(1/2)$.
            \item Encoder 2 generates $Y^n=Y^n(\Wt_n,X^n)=Y^n(\Wt_n,M)$. He only needs the bits $X_i$ such that $\Yt_i=d$, which are available via $M$.
        \end{enumerate}

        To complete the proof, note that $X^n\to (\Wt_n,M)\to Y^n$ forms a Markov chain. Therefore,
        \begin{align*}
            G(X^n;Y^n) & \le H(\Wt_n,M)+1 \le H(\Wt_n)+H(M)+1\\
            &\stackrel{(a)}\le H(\delta(\e))+(1-\delta(\e))H(\Wt_n|\Wt_n\in \Cc)\\
	    &~~~+\delta(\e) H(\Wt_n|\Wt_n\notin \Cc)+n(1-p_1+\delta(\e))+1\\
            &\stackrel{(b)}{\le} H(\delta(\e))+(1-\delta(\e))\log |\Cc| \\
	    &~~~~~~~~~~~+ \delta(\e) \log |\tilde{\Yc}^n| + n(1-p_1+\delta(\e))+1\\
            &= n(I(\Yt;\Wt)+1-p_1+\delta(\e))\\
            &\stackrel{(c)}{=} n(I(W;X,Y)+\delta(\e)),
        \end{align*}
        where $(a)$ follows by the grouping lemma for entropy, since $\P\{\Wt_n\notin \Cc\}=\P\{\Yt^n\text{ not covered}\}=\delta(\e)$; $(b)$ follows since entropy is upper bounded by $\log$ of the alphabet size; and $(c)$ follows from the definition of mutual information and some algebraic manipulations.

        If we let $n\to \infty$, we obtain $\Gbar(X;Y)\le I(W;X,Y)+\delta(\e)$ for any $\e>0$. Minimizing $I(W;X,Y)$ over all $W$ from Lemma \ref{clm:sbesWchoice} completes the proof.
	\end{proof}
\smallskip

Note that the single letter characterization of the Wyner common information for the 2-DMS $(X^k,Y^k)\sim\prod_{i=1}^k p_{X,Y}(x_i,y_i)$ is $k$ times that of the 2-DMS $(X;Y)$, that is, $\min I(W;X^k,Y^k) = k \min I(W;X,Y)$. The same property holds for the G\'{a}cs--K\"{o}rner--Witsenhuesen common information~\cite{Gacs--Korner1973}, and for the mutual information. In the following we show that $G(X^k;Y^k)$ can be strictly smaller than $kG(X;Y)$. Hence, it is possible to realize gains in the ``common entropy" when we increase the dimension.

By the fact that for the SBEC, $\Gbar(X;Y)=H(p)$ for $p>1/2$ and $G(X;Y)= \min \{1, H(p)+1-p\}$, there exists a $p$ such that $\Gbar(X;Y)<G(X;Y)$. Hence, we can show by contradiction that there exists a 2-DMS $(X,Y)$ such that $G(X^2;Y^2)<2G(X;Y)$.
We can also give an explicit example of a 2-DMS $(X,Y)$ such that $G(X^2;Y^2)<2G(X;Y)$. Let
    \[
        p_{X,Y}=
        \begin{bmatrix}
            1/3 & 1/3 \\
            1/3 & 0
        \end{bmatrix}.
    \]
    Then, by Proposition~\ref{prop:2letter}, we have $G(X;Y)=H(1/3)$, where $H(p)$, $0 \le p \le 1$, is the binary entropy function. Now,
    \[
    p_{X^2,Y^2}=
    \begin{bmatrix}
        1/9 & 1/9 & 1/9 & 1/9 \\
        1/9 & 0 & 1/9 & 0 \\
        1/9 & 1/9 & 0 & 0 \\
        1/9 & 0 & 0 & 0
    \end{bmatrix}.
    \]
    If $X^2\to W\to Y^2$, we can write
    \begin{align*}
        p_{X^2,Y^2}(x^2,y^2)&=\sum_w p_W(w)p_{X^2|W}(x^2|w)p_{Y^2|W}(y^2|w).
        \end{align*}
        Therefore, if we write
        \begin{align*}
            p_{X^2,Y^2}&=
            \frac{4}{9}
            \begin{bmatrix} 1/4 \\1/4\\1/4\\1/4 \end{bmatrix}
            \begin{bmatrix} 1 \\0 \\0 \\0  \end{bmatrix}^t
            +\frac{3}{9}
            \begin{bmatrix} 1 \\0 \\0 \\0  \end{bmatrix}
            \begin{bmatrix} 0 \\1/3\\1/3\\1/3 \end{bmatrix}^t
        	+\frac{1}{9}
        	\begin{bmatrix} 0 \\1 \\0 \\0  \end{bmatrix}
            \begin{bmatrix} 0 \\0 \\1 \\0  \end{bmatrix}^t
            +\frac{1}{9}
            \begin{bmatrix} 0 \\0 \\1 \\0  \end{bmatrix}
            \begin{bmatrix} 0 \\1 \\0 \\0  \end{bmatrix}^t.
        \end{align*}
    we can identify $p_W(w)=\begin{bmatrix}4/9,3/9,1/9,1/9\end{bmatrix}$.
    It is easy to see $p_{W_1,W_2}(w_1,w_2)=\begin{bmatrix}4/9,2/9,2/9,1/9\end{bmatrix}$.
    Thus,
    \[
        G(X^2;Y^2)= H(W) < H(W_1,W_2)= 2H(W_1)=2G(X;Y).
    \]
	\subsection{Approximate common information rate}
Consider the approximate distributed generation setting in which Alice and Bob wish to generate 2-DMS $(X,Y)$ with vanishing total variation
	\begin{equation*}
		\lim_{n\to \infty} \big|p_{\hat X^n,\hat Y^n}(x^n,y^n) - \prod_{i=1}^n p_{X,Y}(x_i,y_i)\big|_{\mathrm{TV}} = 0.
	\end{equation*}
We define a $(W_n,R,n)$-simulation code for this setting in the same manner as for exact distributed generation.
We define the {\em approximate common information rate} $R_{\mathrm{TV}}^*$ between the sources $X$ and $Y$ as the infimum over all rates $R$ such that the 2-DMS $(X,Y)$ can be approximately generated.

We can show that the approximate common information rate is equal to the Wyner common information.
	\begin{proposition}\label{prop:wyner}
		\[
			R^*_{\mathrm{TV}} = J(X;Y).
		\]
    \end{proposition}
	\begin{proof}
		\noindent{\emph{Achievability:}} Achievability follows from Wyner's coding scheme~\cite{Wyner1975a}. Choose $W_n\sim \U [1:2^{nR}]$ and associate each $w_n\in\Wc_n$ with a codeword of fixed length $\ell(w_n)=\lceil nR \rceil$. Decoders 1 (Alice) and  2 (Bob) first decode $W_n$ and then use Wyner's coding scheme to generate $\hat{X}^n,\hat{Y}^n$, respectively. Any rate $R>J(X;Y)$ is admissible and will guarantee the existence of a scheme such that $(\hat{X}^n,\hat{Y}^n)$ is close in total variation to $(X^n,Y^n)$. Thus $R^*_{\mathrm{TV}} \le J(X;Y)$.
		
		\noindent{\emph{Converse:}} Suppose that for any $\e>0$, there exists a $(W_n,R,n)$ simulation code that generates $(\hat X^n,\hat Y^n)$ whose pmf differs from that of $(X^n,Y^n)$ by at most $\e$ in total variation. Then we have
		\begin{align*}
			nR &\ge H(W_n) \ge I(\hat X^n,\hat Y^n;W_n)\\
            &= \sum_{q=1}^n I(\hat X_q,\hat Y_q;W|\hat X^{q-1},\hat Y^{q-1})\\
            &= \sum_{q=1}^n I(\hat X_q,\hat Y_q;W,\hat X^{q-1},\hat Y^{q-1})\\
	    &~~~~~~~~~~~~~~~-I(\hat X_q,\hat Y_q;\hat X^{q-1},\hat Y^{q-1})\\
            &\stackrel{(a)}{\ge} \sum_{q=1}^n I(\hat X_q,\hat Y_q;W) - n\delta(\e)\\
            &=nI(\hat X_Q,\hat Y_Q;W|Q)- n\delta(\e)\\
            &=nI(\hat X_Q,\hat Y_Q;W,Q)- nI(\hat X_Q,\hat Y_Q;Q)- n\delta(\e)\\
            &\stackrel{(b)}{\ge} nI(\hat X_Q,\hat Y_Q;W,Q) - n\delta(\e)\\
            &\stackrel{(c)}{\ge} nJ(X;Y) - n\delta(\e).
		\end{align*}
		$(a),(b)$ follow from Lemma 20 and Lemma 21 respectively in~\cite{cuffthesis} since the pmf of $(\hat X^n,\hat Y^n)$ differs from that of $(X^n,Y^n)$ by at most $\e$ in total variation;
		and $(c)$ follows from the continuity of $J(X;Y)$ ~\cite{Wyner1975a}.	\end{proof}
\smallskip
	
	\noindent{\bf Remark}: Note that if we replace the total variation constraint in Proposition \ref{prop:wyner} by the stronger condition
	\begin{equation}\label{eqn:sepdist}
		p_{X^n,Y^n}(x^n,y^n)=(1-\e)p_{\hat X^n,\hat Y^n}(x^n,y^n)+\e r(x^n,y^n)
	\end{equation}
	for some pmf $r(x^n,y^n)$ over $\Xc^n\times \Yc^n$, the required approximate common information rate $R^*_{\mathrm{SD}}$ becomes equal to the exact common information $\Gbar(X;Y)$.
	To show this, note that
	$R^*_{SD}\le \Gbar(X;Y)$ is trivial because the exact distributed generation constraint is stronger than \eqref{eqn:sepdist}.
	
	 To show $R^*_{\mathrm{SD}}\ge \Gbar(X;Y)$, start with any $(W_n,R,n)$ simulation code that generates $(\hat X^n,\hat Y^n)$ satisfying \eqref{eqn:sepdist}. Let
	 \[
		W_n'=\begin{cases}
			W_n&\text{w.p. }1-\e,\\
			(\bar X^n,\bar Y^n)\sim r(x^n,y^n) &\text{w.p. } \e.
		\end{cases}
	 \]
	 We construct a $(W_n',R',n)$ code that generates $(X^n,Y^n)$ exactly and satisfies $R'\le R+\delta(\e)$. If the decoders receive $W_n'=W_n$, they follow the original achievability scheme to generate $(\hat X^n,\hat Y^n)$ satisfying \eqref{eqn:sepdist}. If $W_n'=(\bar X^n,\bar Y^n)$, then the decoders simply output $\bar X^n$ and $\bar Y^n$, respectively. Now,
	\begin{align*}
		H(W_n') &\le H(\e)+ (1-\e)H(W_n)+\e \log |\Xc|^n|\Yc|^n \\
			&=H(W_n)+n\delta(\e).
	\end{align*}
	Therefore, $R'\le (1/n)(H(W_n')+1)=R+\delta(\e)+1/n =R+\delta(\e)$ for $n$ large enough. Thus $R^*_{SD}\ge \Gbar(X;Y)$.
	
    \section{Exact coordination capacity}\label{sec:cuff}
    In this section, we consider exact channel simulation, an extension of channel simulation with total variation constraint introduced in~\cite{cuff2}. Consider the setup shown in Figure \ref{fig:cuff}. Nature generates $X^n\sim \prod_{i=1}^n p_{X}(x_i)$ that is available to the encoder. Both encoder and decoder have access to common randomness $W_n$. The encoder sends a message $M(X^n,W)$ to the decoder. The decoder outputs $\hat Y^n$ using the message $M$, the common randomness $W_n$, and local randomness. We wish to characterize the trade-off between the amount of common randomness and the information rate.

    \begin{figure}[h]
    	\begin{center}
     	\psfrag{w}[b]{$W_n$}
        \psfrag{m}[b]{$M$}
     	\psfrag{a}[c]{Encoder}	
    	\psfrag{b}[c]{Decoder}
    	\psfrag{u}[t]{$X^n$}
    	\psfrag{v}[t]{\,$\hat Y^n$}
    	\includegraphics[scale=0.55]{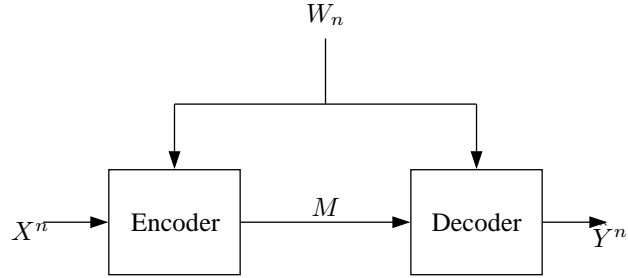}
	   \vspace{5pt}
        \caption{Setting for exact channel simulation.}
        \label{fig:cuff}
    	\end{center}
    \end{figure}

    Formally, a $(W_n,R,R_0,n)$ channel simulation code for this setup consists of
    \begin{itemize}
        \item a common random variable $W_n \sim p_{W_n}(w)$ independent of the source $X^n$. As a measure of the amount of common randomness, we use the per-letter {\em minimum expected codeword length} $R_0$ over the set of all variable length {\em prefix-free} zero-error binary codes $\Cc_0 \subset \{0,1\}^*$ for $W_n$, i.e., $R_0= \min_{\Cc_0} \E(L)$, where $L$ is the codeword length of the code $\Cc_0$ for $W_n$,
        \item an encoding function $M(X^n,W_n)$ that maps $(X^n,W_n)$ into a random variable $M$. As a measure of the information rate, we use the per-letter {\em minimum expected codeword length} $R$ over the set of all variable length {\em prefix-free} zero-error binary codes $\Cc \subset \{0,1\}^*$ for $M$, i.e., $R= \min_{\Cc} \E(L')$, where $L'$ is the codeword length of the code $\Cc$ for $M$, and
        \item a stochastic decoder $p_{\hat Y^n|M,W_n}(y^n|m,w)$ that outputs $\hat Y^n$.
    \end{itemize}

    The channel simulation code is said to simulate the DMC $p_{Y|X}(y|x)$ exactly if $p_{\hat Y^n|X^n}(y^n|x^n)=\prod_{i=1}^np_{Y|X}(y_i|x_i)$.

    We wish to characterize the set of all achievable rates $(R,R_0)$ for which the DMC $p_{Y|X}(y|x)$ can be simulated exactly.

     We do not know the rate region for exact simulation of an arbitrary DMC. In the following, we show that for the erasure channel, it is equal to the rate region under total variation in~\cite{channelsynth}.

    \begin{theorem}
        When $X$ is a binary symmetric source and $p_{Y|X}(y|x)$ is a binary erasure channel with erasure probability $p$, the rate region for exact channel simulation is the set of rate pairs $(R,R_0)$ such that
        \begin{equation}\label{eqn:cuff1}
            \begin{aligned}
                R&\ge r, \\
                R+R_0&\ge H(p) + r\big(1-H((1-p)/r)\big),
            \end{aligned}
        \end{equation}
        for some $1-p\le r \le \min\{2(1-p),1\}$.
    \end{theorem}
    \begin{proof}
        The converse holds for total variation constraint and should therefore trivially hold for the more stringent exact channel simulation constraint.
		The achievability proof closely resembles that of Theorem \ref{thm:SBES}. The central node still generates $\Wt_n$, but the message $M$, instead of being generated at the central node, is now generated by the encoder (Alice). Thus, a rate pair $(R,R_0)$ is achievable if
        \begin{equation}\label{eqn:cuff2}
            \begin{aligned}
                R   &= (1/n)(H(M)+1)+\delta(\e)=1-p_1+\delta(\e), \\
                R_0 &= (1/n) (H(\Wt_n)+1)+\delta(\e) = I(\Yt;\Wt)+\delta(\e) = H(p)-(1-p_1)H(p_2)+\delta(\e).
            \end{aligned}
        \end{equation}
        for some $p_1$ and $p_2$ such that $p=p_1+p_2-p_1p_2$.
        Letting $r=1-p_1$  shows that $R=r,R_0=H(p)-rH((1-p)/r)$ is achievable. Finally note that both the encoder (Alice) and the central node can share the responsibility of generating and sending $\Wt_n$ to the decoder (Bob). Thus an arbitrary fraction of $R_0$ can be removed and instead added to $R$. This shows the equivalence of \eqref{eqn:cuff1} and \eqref{eqn:cuff2}.

    \end{proof}

    \section{Computing $G(X;Y)$}\label{sec:computingG}
    The optimization problem for determining $G(X;Y)$ is in general quite difficult, involving the minimization of a concave function over a complex Markovity constraint. In this section we provide some results on this optimization problem. We provide two bounds on the cardinality of $W$, establish two useful extremal lemmas, and use these results to analytically compute $G(X;Y)$ for binary alphabets. We then briefly discuss a connection to a problem in machine learning.

    We first establish the following upper bound on cardinality.
    \begin{proposition}\label{prop:cardbnd}
        To compute $G(X;Y)$ for a given pmf $p_{X,Y}(x,y)$, it suffices to consider $W$ with cardinality $|\Wc|\le |\Xc||\Yc|$.
    \end{proposition}
    The proof of this proposition is in Appendix~\ref{proof:cardinality}.

    We now state an extremal lemma regarding the optimization problem for $G(X;Y)$ that will naturally lead to another cardinality bound.

    \begin{lemma}\label{lem:distinctsupport}
        Given $p_{X,Y}(x,y)$, let $W$ attain $G(X;Y)$. Then for $w_1\neq w_2$, the supports of $p_{Y|W}(\cdot|w_1)$ and $p_{Y|W}(\cdot|w_2)$ must be different.
    \end{lemma}

    The proof of this lemma is in Appendix~\ref{proof:distinctsupport}.

    Lemma~\ref{lem:distinctsupport} yields the following bound on the cardinality of $W$.
    \begin{proposition}\label{prop:CardBound2}
        To compute $G(X;Y)$ for a given pmf $p_{XY}(x,y)$, it suffices to consider $W$ with cardinality $|\Wc|\le 2^{\min(|\Xc|,|\Yc|)}-1$.
    \end{proposition}
    \begin{proof}
        Suppose $|\Wc|>2^{|\Yc|}-1$. Since there are only $2^{|\Yc|}-1$ non-empty subsets of $\Yc$, by pigeon-hole principle, there exists $w_1\neq w_2$ such that the supports of $p_{Y|W}(\cdot|w_1)$ and $p_{Y|W}(\cdot|w_2)$ are the same. This contradicts Lemma \ref{lem:distinctsupport}. Hence $|\Wc|\le 2^{|\Yc|}-1$. By a symmetric argument, $|\Wc|\le 2^{|\Xc|}-1$.
    \end{proof}
    The following shows that the bound in Proposition \ref{prop:CardBound2} is tight.
\smallskip

    \noindent{\bf Example 2}
    Let $(X,Y)$ be a SBES with $p=0.1$. Since $p_{X,Y}(0,1)=p_{X,Y}(1,0)=0$, the Markovity constraint $X\to W\to Y$ implies that the only $W$ with $|\Wc|=2$ is $W=X$; see \cite{channelsynth}, Appendix A. Hence, $G(X;Y)\le H(X)=1$. However, $H(Y)=H(0.1)+0.1<1$. Thus, the optimal $W^*$ that achieves $G(X;Y)$ requires $|\Wc^*|=3$, making the bound in Proposition~\ref{prop:CardBound2} tight.

    The following is another extremal property of $G(X;Y)$.
    \begin{proposition}\label{prop:extremal}
        Suppose $W$ attains $G(X;Y)$.
        Consider a non-empty subset $\Wc'\subseteq \Wc$. Let $(X',Y')$ be defined by the joint pmf
        \[
            p_{X',Y'}(x,y) = \sum_{w\in \Wc'} \frac{p_W(w)}{\sum_{w'\in \Wc'}p_W(w')} p_{X|W}(x|w)p_{Y|W}(y|w).
        \]
        Then $H(X';Y')=H(W|W\in \Wc')$.
    \end{proposition}
    The proof of this proposition is in Appendix~\ref{proof:extremal}.

    We now use the above results to analytically compute $G(X;Y)$ for binary alphabets, i.e., when $|\Xc|=|\Yc|=2$.
    \begin{proposition}\label{prop:2letter}
        Let $X\sim \bern(p)$ and
        \[
            p_{Y|X}=
            \begin{bmatrix}
                \alpha & \beta \\
                \bar{\alpha} & \bar{\beta}\\
            \end{bmatrix}
        \]
        for some $\a,\b \in[0,1], \bar{\a}=1-\a, \bar{\b}=1-\b$. Let $W$ achieve $G(X;Y)$. Then either
        \begin{align*}
            p_{Y|W}&=
            \begin{bmatrix}
                \alpha & 1 \\
                \bar{\alpha} & 0  \\
            \end{bmatrix},\;\; p_{W|X}=
            \begin{bmatrix}
                    1 & \bar{\beta}/\bar{\alpha} \\
                    0 & 1-\bar{\beta}/\bar{\alpha}
            \end{bmatrix},\text{ and}\\
            W&\sim \bern\left(\bar{p}\left(1-\bar{\beta}/\bar{\alpha}\right)\right),
        \end{align*}
        or
        \begin{align*}
            p_{Y|W}&=
            \begin{bmatrix}
                0& \beta \\
                1 & \bar{\beta}
            \end{bmatrix},\;\;p_{W|X}=
            \begin{bmatrix}
                1-\alpha/\beta & 0 \\
                \alpha/\beta & 1
            \end{bmatrix},\text{ and}\\
            W&\sim \bern\left(p(1-\alpha/\beta)\right).
        \end{align*}
    \end{proposition}
    The proof of this proposition uses Lemma~\ref{lem:distinctsupport} as well as the cardinality bound $|\Wc|\le 3$ derived from Proposition \ref{prop:CardBound2}. It considers all possible cases for $W$ and finally concludes that $|\Wc|=2$ suffices. The detailed arguments can be found in Appendix~\ref{proof:2letter}.
\smallskip

\noindent{\bf Remark} (Relationship to machine learning): Computing $G(X;Y)$  is closely related to \emph{positive matrix factorization}, which has applications in recommendation systems, e.g.,~\cite{netflix}. In that problem, one wishes to factorize a matrix $M$ with positive entries in the form $M=AB$, where $A$ and $B$ are both matrices with positive entries. Indeed, finding a Markov chain $X\to W\to Y$ for a fixed $p_{X,Y}$ is akin to factorizing $p_{Y|X}=p_{Y|W}p_{W|X}$ and numerical methods such as in~\cite{algofactor} can be used. Rather than minimizing the number of factors $|\Wc|$ as is done in positive matrix factorization literature, it may be more meaningful for recommendation systems to minimize the entropy of the factors $W$. Computing $G(X;Y)$ for large alphabets appears to be very difficult, however.

\section{Conclusion}
We introduced the notion of exact common information for correlated random variables $(X,Y)$ and bounded it by the common entropy quantity $G(X;Y)$. For the exact generation of a 2-DMS, we established a multiletter characterization of the exact common information rate. While this multiletter characterization is in general greater than or equal to the Wyner common information, we showed that they are equal for the SBES. The main open question is whether the exact common information rate has a single letter characterization in general. Is it always equal to the Wyner common information? Is there an example 2-DMS for which the exact common information rate is strictly larger than the Wyner common information? It would also be interesting to further explore the application to machine learning.

\section{Acknowledgments}
The authors are indebted to Young-Han Kim for comments that greatly improved the presentation of this paper.
\appendices
\section{Computing $G(X;Y)$ for SBES}\label{proof:sbesH}
    As described in \cite{channelsynth}, Appendix A, for any $W$ satisfying $X\to W\to Y$, each $w\in \Wc$ must fall in one of the three categories:
    \begin{enumerate}
        \item If $p_{X|W=w}p_{Y|W=w}$ is positive for $Y=0$, then $X=0$ with probability $1$.
        \item If $p_{X|W=w}p_{Y|W=w}$ is positive for $Y=1$, then $X=1$ with probability $1$.
        \item The last category is when $p_{X|W=w}p_{Y|W=w}$ is zero on  $Y\in\{0,1\}$, i.e., $p_{Y|W}(e|w)=1$.
    \end{enumerate}
    Following the same reasoning in~\cite{channelsynth}, Appendix A, we conclude that we need only one $w$ in each category. Thus,

    \begin{align*}
        p_{X,Y} &=
            \begin{bmatrix}
                (1-p)/2 & 0 & p/2\\
                0& (1-p)/2 &p/2
            \end{bmatrix}\\
        &=\sum_{w\in \Wc}p_W(w)p_{X|W}(\cdot|w)p_{Y|W}(\cdot|w)^t\\
        &=  \begin{bmatrix}
                a_1 & 0 & a_2\\
                0   & 0 & 0
            \end{bmatrix}+
            \begin{bmatrix}
                0 & 0   & 0\\
                0 & b_1 & b_2
            \end{bmatrix}+
            \begin{bmatrix}
                0 & 0 & c_1\\
                0 & 0 & c_2
            \end{bmatrix}.
    \end{align*}
    Each term above corresponds to a category of $w$.
    Thus,
    \[
    p_W=\begin{bmatrix} a_1+a_2 & b_1+b_2 & c_1+c_2 \end{bmatrix}
    \]
    and we must minimize $H(W)$ such that $a_1=b_1=(1-p)/2$, $a_2+c_1=b_2+c_2=p/2$, and $b_1,b_2,c_1,c_2 \ge 0$. This results in
    \[
        G(X;Y)=\min_{0\le c_1,c_2\le p/2} H\left(1/2-c_1,1/2-c_2,c_1+c_2\right).
    \]
    Since the objective is a concave function of $c_1$ and $c_2$, the minimizing $(c_1,c_2)$ should be one of the 4 corner points $(0,0)$, $(p/2,p/2)$, $(0,p/2)$, $(p/2,0)$. The last two points are symmetric. Thus
    \[
        G(X;Y)=\min \{1,H(p)+1-p \}.
    \]

\section{Subadditivity of $G(X^n;Y^n$)}\label{proof:subadditivity}
    Suppose $W_m$ achieves $G(X^m;Y^m)$ and $W_n$ achieves $G(X^n;Y^n)$. Then, $X^m\to W_m\to Y^m$ and $X_{m+1}^{m+n}\to W_n\to Y_{m+1}^{m+n}$ form Markov chains, and so does $X^{m+n}\to (W_m,W_n)\to Y^{m+n}$. Therefore,
    \begin{align*}
        G(X^{m+n};Y^{m+n}) &\le H(W_m,W_n)\\
        &=H(W_m)+H(W_n)\\
        &=G(X^m;Y^m)+G(X^n;Y^n).
    \end{align*}
    Thus the sequence $\{G(X^n;Y^n):\, n\in\mathbb{N}\}$ is sub-additive. Hence, an appeal to Feteke's subadditivity lemma~\cite{fetekeLemma} shows that
    \[
        \lim_{n\rightarrow\infty} \frac{1}{n} G(X^n;Y^n)=\inf_{n \in \mathbb{N}} (1/n)G(X^n;Y^n).
    \]
\section{Proof of Proposition~\ref{prop:bound}} \label{proof:greater}
To show this consider
\begin{align*}
                \Gbar(X;Y)&=\lim_{n\to\infty} \min_{W:X^n\to W\to Y^n} {1\over n}H(W)\\
                &=\lim_{n\to\infty} \frac{1}{n} H(W_n^*)\\
                &\ge \lim_{n\to\infty} \frac{1}{n} I(W_n^*;X^n,Y^n)\\
                &\ge \lim_{n\to\infty} \frac{1}{n} \sum_{i=1}^n I(W_n^*;X_i,Y_i)\\
                &\ge \min_{W:X-W-Y} I(W;X,Y)\\
                &=J(X;Y).
            \end{align*}
%
\section{Proof of Proposition \ref{prop:cardbnd}}\label{proof:cardinality}
    We use the perturbation method (see Appendix C in~\cite{abbasbook}). The Markov Chain $X\to W \to Y$ is equivalent to
    \begin{align*}
        p(y|x,w) &= p(y|w),\text{ or} \\
        \frac{p(x,y,w)}{\sum_{y'} p(x,y',w)}&=\frac{\sum_{x'} p(x',y,w)}{\sum_{x',y'}p(x',y',w)}.
    \end{align*}
    Further, $\sum_w p(x,y,w)=p_{X,Y}(w)$.

    Let $p_\epsilon(x,y,w)=p(x,y,w)(1+\epsilon\phi(w))$ be a perturbed pmf, where $\epsilon$ can be either positive or negative. We first observe that
    \begin{align*}
        p_\epsilon(y|x,w) &= \frac{p_\epsilon(x,y,w)}{\sum_{y'} p_\epsilon(x,y',w)} \\
        &= \frac{p(x,y,w)}{\sum_{y'} p(x,y',w)} \\
        &=\frac{\sum_{x'}p(x',y,w)}{\sum_{x',y'}p(x',y',w)}  \\
        &=\frac{\sum_{x'}p_\epsilon(x',y,w)}{\sum_{x',y'}p_\epsilon(x',y',w)}  \\
        &=p_\epsilon(y|w).
    \end{align*}
    Thus, $p_\epsilon(x,y,w)$ also satisfies the Markov chain $X\to W \to Y$.

    Now, we also require that $p_\epsilon(x,y)=p_{X,Y}(x,y)$, i.e.,
    \begin{equation*}
        \sum_w p(x,y,w)\phi(w)=0 \text{ for all } (x,y).
    \end{equation*}
    The above equation represents $|\Xc||\Yc|$ linear constraints in $\phi(w)$. Thus if $|\mathcal{W}|>|\mathcal{X}||\mathcal{Y}|$, we can find a perturbation $\phi(w)\neq 0$ such that $p_\epsilon(x,y,w)$ satisfies $\sum_w p_\epsilon(x,y,w)=p_{X,Y}(x,y)$ as well as the Markov chain $X\to W \to Y$.

    If we choose $\e$ small enough, we can ensure that $1+\e\phi(w)\geq 0$ for all $w$ and thus $p_\epsilon(u,v,w)$ is a valid pmf.

    Note that $p_\epsilon(w)=p(w)(1+\epsilon\phi(w))$ is a linear function of $\e$. Since the entropy is a concave function, $H(p_\epsilon(w))$ is a concave function of $\epsilon$ and is minimized by an extremal $\epsilon$. Such an extremal $\epsilon$ makes $p_\epsilon(w)=0$ for some $w$, thus reducing the cardinality $|\mathcal{W}|$ by $1$ while also reducing $H(W)$. Thus we conclude for the minimizing $H(W)$, $|\mathcal{W}|\le|\mathcal{X}||\mathcal{Y}|$.
%
\section{Proof of Lemma \ref{lem:distinctsupport}}\label{proof:distinctsupport}
    We will prove by contradiction. Assume that Lemma \ref{lem:distinctsupport} is false, i.e. $\exists w_1\neq w_2$ such that the supports of $p_{Y|W}(\cdot|w_1)$ and $p_{Y|W}(\cdot|w_2)$ are the same for some $W$ that achieves $H_1(X;Y)$.{\allowdisplaybreaks

    Let $\Xc=\{1,2,\dots,|\Xc|\}$, $\Yc=\{1,2,\dots,|\Yc|\}$, etc. without loss of generality.

    For a given $w$, $p_{Y|W}(\cdot|w_1)$ is a vector in $\real^{|\Yc|}$ whose $y$-th element is the conditional probability $p_{Y|W}(y|w)$. The line segment joining $p_{Y|W}(\cdot|w_1)$ and $p_{Y|W}(\cdot|w_2)$ can be extended from both ends without crossing the boundary of the $|\Yc|-1$ dimensional probability simplex as shown in Fig. \ref{fig:higherdim}.

    Geometrically, the Markov chain $X\to W\to Y$ means that the vectors $\{p_{Y|X}(\cdot|x):x\in \Xc\}$ lie in the convex hull of the vectors $\{p_{Y|W}(\cdot|w):w\in \Wc\}$.

    \begin{figure}[h]
    	\begin{center}\footnotesize
    	\psfrag{a}[t]{$p_{Y|W'}(y|w_1)$}
         	\psfrag{b}[b]{$a$}
        	\psfrag{c}[b]{$p_{Y|W}(y|w_1)$}
    	\psfrag{d}[b]{$1$}
         	\psfrag{e}[b]{$p_{Y|W}(y|w_2)$}
        	\psfrag{f}[b]{$b$}
    	\psfrag{g}[t]{$p_{Y|W}(y|w_3)$}
        	\psfrag{h}[t]{$p_{Y|W'}(y|w_2)$}
    	\includegraphics[scale=0.87]{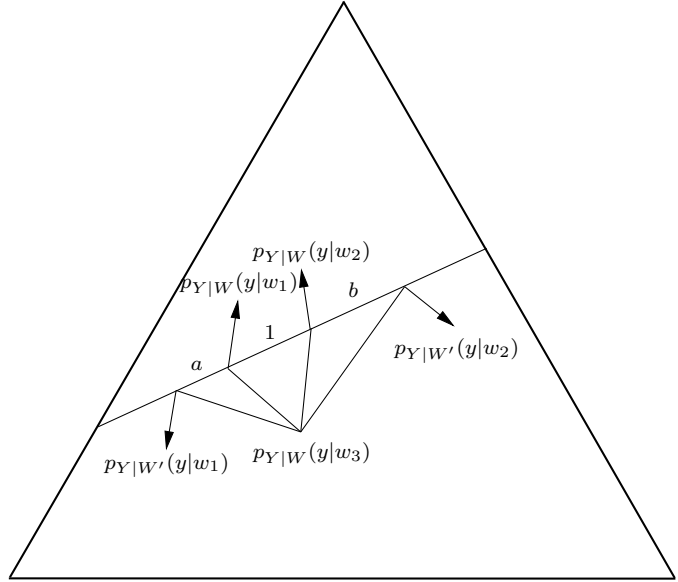}
        \caption{Proof of Lemma \ref{lem:distinctsupport}.}
        \label{fig:higherdim}
    	\end{center}
    \end{figure}
    We will construct $W'$ such that $X\to W\to W'\to Y$ forms a Markov chain and $H(W')<H(W)$. Assume w.l.o.g. that the Euclidean distance between $p_{Y|W}(\cdot|w_1)$ and $p_{Y|W}(\cdot|w_2)$ in $\real^{|\Yc|}$ is $1$ unit. Extend the line segment in either direction by $a,b$ units respectively so that it's new end-points are $p_{Y|W'}(\cdot|w_1),p_{Y|W'}(\cdot|w_2)$ (see Fig. \ref{fig:higherdim}). Then $X\to W\to W'\to Y$ forms a Markov chain, where $W'$ is defined as follows:

    \[
        p_{Y|W}(y|w)=\begin{cases}
        p_{Y|W'}(y|w)&\text{if }w\notin \{w_1,w_2\}, \\[2pt]
        \frac{b+1}{a+b+1}p_{Y|W'}(y|w_1)+\frac{a}{a+b+1}p_{Y|W'}(y|w_2) &\text{if } w=w_1, \\[2pt]
        \frac{b}{a+b+1}p_{Y|W'}(y|w_1)+\frac{a+1}{a+b+1}p_{Y|W'}(y|w_2) &\text{if } w=w_2.
        \end{cases}
    \]
    and
    \[
        p_{W'|W}(w'|w)=\begin{cases}
        \one(w'=w) &\text{if } w\notin \{w_1,w_2\}, \\[2pt]
        \frac{b+1}{a+b+1} &\text{if } (w',w)=(w_1,w_1), \\[2pt]
        \frac{a}{a+b+1} &\text{if } (w',w)= (w_2,w_1),\\[2pt]
        \frac{b}{a+b+1} &\text{if } (w',w)= (w_1,w_2),\\[2pt]
        \frac{a+1}{a+b+1} &\text{if } (w',w)= (w_2,w_2),\\[2pt]
        0& \text{otherwise.}
        \end{cases}
    \]

    Hence we have
    \[
    p_{W'}(\cdot) = \left[\frac{b+1}{a+b+1}p_{W}(w_1)+\frac{b}{a+b+1}p_{W}(w_2),\frac{a}{a+b+1}p_{W}(w_1)+ \frac{a+1}{a+b+1}p_W(w_2),p_W(3),p_W(4),\dots \right],
    \]
    and
    \begin{equation}
    \label{eqn:Hextreme}
    H(W)-H(W')=\P\{W\in\{w_1,w_2\}\}\left[H\left(p\right)-H\left(\frac{b+p}{a+b+1}\right)\right],
    \end{equation}
    where $p=\P\{W=w_1|W\in\{w_1,w_2\}\}$ does not depend on $a,b$ and $\bar{p}=1-p$. In~\eqref{eqn:Hextreme}, the term $(b+p)/(a+b+1)$ increases with $b$ and decreases with $a$. Further, when $a=b=0$, $(b+p)/(a+b+1)=p$. Hence, by perturbing either $a$ or $b$, one can make the RHS of \eqref{eqn:Hextreme} positive. In that case $H(W')<H(W)=G(X;Y)$, a contradiction.
%
\section{Proof of Proposition \ref{prop:extremal}}\label{proof:extremal}
    Define the random variable $W'\in\Wc'$ with pmf $p_{W'}(w)=\frac{p_W(w)}{\sum_{w'\in \Wc'}p_W(w')}$. Then $X'\to W'\to Y'$ and $H(W')=H(W|W\in \Wc')$. Hence $H(X';Y')\le H(W')= H(W|W\in \Wc')$.

    Now we show by contradiction that $H(X';Y')= H(W|W\in \Wc')$.

    If $H(X';Y') < H(W|W\in \Wc')$, there exists an a $\Wt\in \tilde{\Wc}$ such that $X'\to \Wt\to Y'$ and $H(\Wt)<H(W')$. We have
    \begin{align*}
        p_{X,Y}(x,y)&=\sum_{w\in\Wc} p_W(w)p_{X|W}(x|w)p_{Y|W}(y|w)\\
        &=\sum_{w\in \Wc'} p_W(w)p_{X|W}(x|w)p_{Y|W}(y|w) + \sum_{w\in \Wc \backslash \Wc'} p_W(w)p_{X|W}(x|w)p_{Y|W}(y|w)\\
        &=\P\{W\in \Wc'\}p_{X',Y'}(x,y) + \sum_{w\in \Wc\backslash\Wc'} p_W(w)p_{X|W}(x|w)p_{Y|W}(y|w)\\
        &=\P\{W\in \Wc'\}\sum_{w\in \tilde{\Wc}}p_{\tilde{W}}(w)p_{X'|\tilde{W}}(x|w)p_{Y'|\tilde{W}}(y|w) + \sum_{w\in \Wc\backslash\Wc'} p_W(w)p_{X|W}(x|w)p_{Y|W}(y|w).
    \end{align*}
    Thus, if we define
    \[
        W''=
        \begin{cases}
            W,& \text{if } W\in \Wc',\\
            \Wt,& \text{if } W\notin \Wc'.
        \end{cases}
    \]
    Then $X\to W''\to Y$ forms a Markov chain. Also
    \begin{align*}
        H(W'')&\stackrel{(a)}{=}H\left(\P\{W\in \Wc'\}\right)+\P\{W\in \Wc'\}H(W')+\P\{W\notin \Wc'\}H(W|W\notin \Wc')\\
        &<H\left(\P\{W\in \Wc'\}\right)+\P\{W\in \Wc'\}H(\Wt)+\P\{W\notin \Wc'\}H(W|W\notin \Wc')\\
        &\stackrel{(b)}{=}H(W),
    \end{align*}
    where $(a)$ and $(b)$ follow since $H(X,\Theta)=H(\Theta)+H(X|\Theta)$. Thus we obtain $G(X;Y)<H(W)$, a contradiction.}
%
\section{Proof of Proposition \ref{prop:2letter}}\label{proof:2letter}
    Proposition \ref{prop:CardBound2} implies $|\Wc|\le 3$.

    Geometrically, the Markov chain $X\to W\to Y$ means that the vectors $\{p_{Y|X}(\cdot|x):x\in \Xc\}$ lie in the convex hull of the vectors $\{p_{Y|W}(\cdot|w):w\in \Wc\}$.

    When we restrict $|\Wc|=2$, an application of Lemma \ref{lem:distinctsupport} together with the observation that $p_{W|X}(w|x)=0 \iff p_{X|W}(x|w)=0$ and the Markovity $X\to W\to Y$ immediately results in the two cases described in Proposition \ref{prop:2letter}.

    Similarly, when we let $|\Wc|=3$, by Lemma \ref{lem:distinctsupport}, the supports of $p_{Y|W}(\cdot|w)$ are different for each $w\in \Wc$. Therefore, we have
    \[
    p_{Y|W}=\begin{bmatrix}
        1 & 0 & x\\
        0 & 1 & \xb\\
    \end{bmatrix},
    \]
    where each column corresponds to a value of $W$ and $0\le x\le 1$.

    Similarly $p_{W|X}$ can be one of the 6 cases below. Each row corresponds to a $W$. Since there are $3$ rows, there are $3!=6$ cases.
    \[
    p_{W|X}\in \left\{
    \begin{bmatrix}
        y & z \\
        \yb & 0\\
        0 & \zb
    \end{bmatrix},
    \begin{bmatrix}
        y & z \\
        0 & \zb\\
        \yb & 0
    \end{bmatrix},
    \begin{bmatrix}
        0 & \zb \\
        y & z\\
        \yb & 0
    \end{bmatrix},
    \begin{bmatrix}
        \yb & 0 \\
        y & z\\
        0 & \zb
    \end{bmatrix},
    \begin{bmatrix}
        \yb & 0 \\
        0 & \zb\\
        y & z
    \end{bmatrix},
    \begin{bmatrix}
        0 & \zb \\
        \yb & 0\\
        y & z
    \end{bmatrix}
    \right\},
    \]
    for appropriately chosen numbers $y,z$ between 0 and 1.

    Assume
    \[
    p_{Y|X}=\begin{bmatrix}
        \a & \b \\
        \bar{\a} & \bar{\b}
    \end{bmatrix}
    \]
    and $X\sim\bern(p)$.
    We first consider case $1$ for $p_{W|X}$.
    In that case, we have the Markovity $p_{Y|W}p_{W|X}=p_{Y|X}$, i.e.,
    \[
    \begin{bmatrix}
        1 & 0 & x\\
        0 & 1 & \xb\\
    \end{bmatrix}
    \begin{bmatrix}
        y & z \\
        \yb & 0\\
        0 & \zb
    \end{bmatrix}
    =\begin{bmatrix}
        \a & \b \\
        \bar{\a} & \bar{\b}
    \end{bmatrix},
    \]
    which yields $y=\a$ and $z+\zb x=\b$ or $z=\frac{\b-x}{\xb}$, hence $x\le \b$.
    Also
    \[
    p_W=p_{W|X}p_X=\begin{bmatrix}
        py+\pb z \\
        p\yb \\
        \pb \zb
    \end{bmatrix}.
    \]
    $y=\a$ is fixed, and $z$ decreases monotonically with $x$. Since $p_W$ is a linear function of $z$, the optimal $z$ that minimizes the concave function $H(W)$ should lie at one of the end-points of $z$. The two end-points are $z=0$ (corresponding to $x=\b$) and $z=\b$ (corresponding to $x=0$). When $z=0$, $H(W)>H(\pb)=H(X)$ and may be eliminated because $H_1(X;Y)\le H(X)$. When $z=\b$, we have $x=0$, hence $p_{Y|W}(y|1)=p_{y|W}(y|2)$ for all $y$.

    In this case, we can replace $W$ by it's sufficient statistic w.r.t $Y$, namely
    \[
    T(W)=\begin{cases}
    1, & \text{if }$w=0$\\
    0, & \text{otherwise.}
    \end{cases}
    \]
    By property \ref{prop:minsuff} from Section \ref{sec:1-shot} this operation reduces the cardinality $|\Wc|$ and increases the entropy.

    Cases $2,3,4$ for $p_{W|X}$ are very similar to case $1$. Now consider case $5$. In this case, Markovity implies
    \[
    \begin{bmatrix}
        1 & 0 & x\\
        0 & 1 & \xb\\
    \end{bmatrix}
    \begin{bmatrix}
        \yb & 0 \\
        0 & \zb\\
        y & z
    \end{bmatrix}
    =\begin{bmatrix}
        \a & \b \\
        \bar{\a} & \bar{\b}
    \end{bmatrix}
    \]
    and
    \[
    p_W=\begin{bmatrix}
        p\yb \\
        \pb \zb \\
        py+\pb z.
    \end{bmatrix}
    \]
    Thus we want to solve the following optimization problem:
    \begin{align*}
    &\text{Minimize }H(p\yb,\pb \zb,py+\pb z) \\
    &\text{Subject to } \yb+xy=\a, xz=\b.
    \end{align*}
    We can eliminate variable $x$ from the above constraints to obtain a constraint
    \[
        \frac{\bar{\a}}{y}+\frac{\b}{z}=1.
    \]
    We want to show that for the optimal $(y^*,z^*)$, either $y^*=1$ or $z^*=1$. In that case, the support of $W$ reduces to size $2$, thus showing that $|\Wc|=2$ suffices.
    To show this, we observe that $\{(y,z): \frac{\bar{\a}}{y}+\frac{\b}{z}=1\}$ lies in the convex hull of the points
    $\{(1,\b/\a),(\bar{\a}/\bar{\b},1), (0,0),(0,1),(1,0) \}$. Since $H()$ is a concave function, it's minimum in the convex hull should be one of the boundary points. It is easy to see that the minimum must be one of $\{(1,\b/\a),(\bar{\a}/\bar{\b},1)\}$. These 2 points satisfy the original constraint
    $\frac{\bar{\a}}{y}+\frac{\b}{z}=1$. Hence, either $y^*=1$ or $z^*=1$.

    Case $6$ is similar to case $5$.

\bibliographystyle{IEEEtran}
\bibliography{myBib}

\newcommand{\noopsort}[1]{}
\begin{thebibliography}{10}
\providecommand{\url}[1]{#1}
\csname url@samestyle\endcsname
\providecommand{\newblock}{\relax}
\providecommand{\bibinfo}[2]{#2}
\providecommand{\BIBentrySTDinterwordspacing}{\spaceskip=0pt\relax}
\providecommand{\BIBentryALTinterwordstretchfactor}{4}
\providecommand{\BIBentryALTinterwordspacing}{\spaceskip=\fontdimen2\font plus
\BIBentryALTinterwordstretchfactor\fontdimen3\font minus
  \fontdimen4\font\relax}
\providecommand{\BIBforeignlanguage}[2]{{%
\expandafter\ifx\csname l@#1\endcsname\relax
\typeout{** WARNING: IEEEtran.bst: No hyphenation pattern has been}%
\typeout{** loaded for the language `#1'. Using the pattern for}%
\typeout{** the default language instead.}%
\else
\language=\csname l@#1\endcsname
\fi
#2}}
\providecommand{\BIBdecl}{\relax}
\BIBdecl

\bibitem{Wyner1975a}
A.~D. Wyner, ``The common information of two dependent random variables,''
  \emph{{IEEE} Trans. Inf. Theory}, vol.~21, no.~2, pp. 163--179, Mar. 1975.

\bibitem{Ahlswede--Csiszar1993}
R.~Ahlswede and I.~Csisz{\'a}r, ``Common randomness in information theory and
  cryptogra-phy---{I}: {S}ecret sharing,'' \emph{{IEEE} Trans. Inf. Theory},
  vol.~39, no.~4, pp. 1121--1132, 1993.

\bibitem{Cover--El-Gamal--Salehi1980}
T.~M. Cover, A.~El~Gamal, and M.~Salehi, ``Multiple access channels with
  arbitrarily correlated sources,'' \emph{{IEEE} Trans. Inf. Theory}, vol.~26,
  no.~6, pp. 648--657, Nov. 1980.

\bibitem{channelsynth}
P.~Cuff, ``Distributed channel synthesis,'' \emph{arXiv:1208.4415v3}, 2013.

\bibitem{minimalsufficientstatistic}
A.~A. Borokov, ``Mathematical statistics,'' \emph{Gordon and Breach Science
  Publishers}, 1998.

\bibitem{abbasbook}
A.~E. Gamal and Y.~H. Kim, ``Network information theory,'' \emph{Cambridge
  University Press}, 2011.

\bibitem{Gacs--Korner1973}
P.~G{\'a}cs and J.~K{\"o}rner, ``Common information is far less than mutual
  information,'' \emph{Probl. Control Inf. Theory}, vol.~2, no.~2, pp.
  149--162, 1973.

\bibitem{cuffthesis}
P.~Cuff, ``Communication in networks for coordinating behavior,'' \emph{PhD
  Dissertation submitted to the dept. of Electrical Engg. at Stanford
  University}, 2009.

\bibitem{cuff2}
P.~Cuff, H.~Permuter, and T.~Cover, ``Coordination capacity,'' \emph{IEEE
  Transactions on Information Theory}, vol.~56, no.~9, p. 4181 – 4206,
  September 2010.

\bibitem{netflix}
K.~Y, R.~B, and C.~V, ``Matrix factorization techniques for recommender
  systems,'' \emph{IEEE Computer Society}, 2009.

\bibitem{algofactor}
D.~D. Lee and H.~S. Seung, ``Algorithms for non-negative matrix
  factorization,'' \emph{Advances in Neural Information Processing Systems 13:
  Proceedings of the 2000 Conference. MIT Press.}, pp. 556--562, 2001.

\bibitem{fetekeLemma}
J.~Vant and R.~Wilson, ``A course in combinatorics,'' \emph{Cambridge
  University Press}, 1992.

\end{thebibliography}

\end{document}